\documentclass[a4paper,12pt]{article}
\usepackage[utf8]{inputenc}
\usepackage[english]{babel}
\usepackage[left=2cm,right=2cm,top=2cm,bottom=2cm]{geometry}

\usepackage{amsmath,amsfonts,amssymb,amsthm}
\usepackage{bm,color}
\usepackage{placeins}
\usepackage{graphicx}
\usepackage{enumerate}
\usepackage{mathtools}
\usepackage{float}
\usepackage[hidelinks]{hyperref}

\usepackage{pdfpages}
\usepackage{lscape}
\usepackage{supertabular}
\usepackage{comment}

\usepackage{tikz}
\tikzstyle{mybox} = [draw=black, very thick, rectangle, rounded corners, inner ysep=5pt, inner xsep=5pt]

\usepackage{caption}
\usepackage{subcaption}

 \usepackage{listings}
\usepackage{textcomp}

\numberwithin{equation}{section}
\title{\textbf{Ballistic electron transport described by 
a generalized Schr\"{o}dinger equation}}
\date{}
\author{Giulia Elena Aliffi\thanks{University of Catania, Department of Mathematics and Computer Science, Catania, Italy  ({\tt giuliaelena.aliffi@phd.unict.it}).}
\and 
Giovanni Nastasi\thanks{University of Enna "Kore", Department of Engineering and Architecture, Enna, Italy ({giovanni.nastasi@unikore.it}).}
\and Vittorio Romano\thanks{University of Catania, Department of Mathematics and Computer Science, Catania, Italy  (Corresponding author: {\tt vittorio.romano@unict.it}).} }

\newcommand{\R}{\mathbb{R}}
\newcommand{\C}{\mathbb{C}}
\newcommand{\N}{\mathbb{N}}

\newtheorem{proposition}{Proposition}
\newtheorem{lemma}{Lemma}

\newtheorem{remark}{Remark}

\begin{document}

\maketitle

\begin{abstract}
We propose a Schrödinger equation of arbitrary order for modeling charge transport in semiconductors operating in the ballistic regime. This formulation incorporates non-parabolic effects through the Kane dispersion relation, thereby extending beyond the conventional effective mass approximation. Building upon the framework introduced in G.E. Aliffi, G. Nastasi, V. Romano,  {ZAMP} {76}, 155 (2025), we derive a hierarchy of models, each governed by a Schrödinger equation of increasing order.
As in the standard second-order formulation, the problem is formulated on a bounded spatial domain with suitable transparent boundary conditions. These conditions are designed to simulate charge transport in a quantum coupler where an active region --~representing the electron device~-- is connected to leads acting as reservoirs.
We investigate several analytical properties of the proposed models and derive a generalized expression for the current, valid for any order. This formula includes additional terms that account for interference effects arising from the richer wave structure inherent in higher-order Schrödinger equations, which are not captured by the effective mass approximation.
Numerical simulations of a resonant tunneling diode (RTD) illustrate the key features of the solutions and highlight the impact of the generalized formulation on device behaviour.
\end{abstract}

 \vskip 1cm
 {\bf Author Contribution declaration} 
 
 Giovanni Nastasi: Formal analysis, Investigation, Methodology, Supervision, Validation,
Visualization, Writing – original draft, Writing .

Giulia Aliffi: Formal analysis, Investigation, Methodology, Validation,
Visualization, Writing – original draft, Writing .

Vittorio Romano: Conceptualization, Formal analysis, Funding acquisition, Investigation, Methodology, Supervision, Validation, Writing
– original draft, Writing .


\maketitle

\newpage

\section{Introduction}

The ongoing trend of enhanced miniaturization in semiconductor technology has made quantum effects essential for understanding the full range of charge transport phenomena in electron devices, especially those with active regions spanning only a few nanometers and featuring abrupt potential variations \cite{HaVa,Da}. While the semiclassical Boltzmann equation remains highly accurate for devices with characteristic scales on the order of microns, it fails to capture quantum tunneling effects, necessitating the use of fully quantum mechanical models.
A widely studied benchmark for quantum transport is the resonant tunneling diode (RTD), a heterostructure typically composed of alternating layers of GaAs and AlGaAs, which generate potential barriers. Various quantum approaches have been employed to model charge transport in RTDs, including the Wigner equation—often solved using signed Monte Carlo particle methods \cite{NeKoSeRiFe,QuDo,Mu,Ja}—and the non-equilibrium Green’s function formalism \cite{LaDa}.
Under the assumption of ballistic transport, an alternative and direct approach involves solving the Schrödinger equation to simulate electron flow in RTDs \cite{LeKi,AbDeMa,Ab,Ju}. As proposed in \cite{LeKi} (see also \cite{Fre}), the RTD is modeled as being connected to contacts that act as waveguides, allowing electrons to enter and exit the active region. A key idea is the use of transparent boundary conditions, which reduce the problem—originally posed on the entire space—to a boundary value problem defined on a compact domain.
Most existing simulations adopt the standard effective mass approximation and inject electrons according to Fermi–Dirac statistics \cite{MeJuKo,HeWa,Pi,Ar,ArNe}. However, this approximation corresponds to a parabolic band structure, which is known to overestimate the electric current \cite{Ro}. This motivates the inclusion of more accurate dispersion relations in the Schrödinger equation.
The Kane dispersion relation \cite{Lun,CaMaRo_book} offers a refined analytical model that improves upon the parabolic band approximation while avoiding the computational complexity of full-band models \cite{FiLa,FiLabis}, which are typically accessible only through numerical methods.

In \cite{AliNaRo}, a correction of order $\hbar^4$ was introduced into the dispersion relation, leading to a fourth-order Schrödinger equation (SE) whose analytical properties were thoroughly investigated. Simulations involving single and double potential barriers with an applied bias revealed interference effects absent in the second-order SE. In this work, we generalize that approach to arbitrary order in the expansion of the Kane dispersion relation. We derive a continuity equation for the single-electron probability density and a general expression for the probability current valid at any order. Transparent boundary conditions are also formulated for higher-order models.
We investigate several properties of the proposed framework and present numerical examples to illustrate its behavior.
The structure of the paper is as follows. In Section \ref{sec:SE}, we derive the generalized Schrödinger equation and discuss the functional setting. Section \ref{sec:curr} is devoted to the derivation of a general expression for the probability current. In Section \ref{sec:TBC}, we formulate transparent boundary conditions for arbitrary order, enabling the reduction of the problem to a compact domain. Section \ref{sec:well} is dedicated to the well-posedness of the resulting boundary value problem. Section \ref{sec:comparison} compares the probability current for the second- and fourth-order SE. Finally, Section \ref{sec:RTD} presents numerical simulations of an RTD to highlight key features of the proposed models.

\section{Dispersion relation and generalized Schr\"{o}dinger equation} \label{sec:SE}

The general form of the dispersion relation in a semiconductor is derived by solving the single-electron Schrödinger equation under a periodic potential, employing Bloch's theorem \cite{Ja,Ju,Lun}. The complete band structure can only be obtained numerically, for instance using pseudopotential methods \cite{FiLabis,CheCo}. However, in practical applications, analytical approximations are often preferred \cite{Da}. Among these, the Kane dispersion relation is widely used. In its isotropic form, it is given implicitly by
\begin{equation}
\epsilon(k)\left(1 + \alpha \epsilon(k)\right) = \frac{\hbar^2 k^2}{2m^*} := \gamma^2,
\end{equation}
where $\epsilon(k)$ denotes the electron energy measured from the bottom of the conduction band $E_C$, $k$ is the magnitude of the electron wave vector, and $m^*$ is the effective electron mass. For example, in GaAs, $m^* = 0.067 m_e$, with $m_e$ being the free electron mass. The positive parameter $\alpha$ is referred to as the non-parabolicity factor. In the limit $\alpha \to 0^+$, the standard parabolic band approximation is recovered:
\begin{equation}
\epsilon(k) = \gamma^2.
\end{equation}

For $\alpha \neq 0$, the energy can be expressed explicitly as
\begin{equation}
\epsilon = \frac{-1 + \sqrt{1 + 4\alpha \gamma^2}}{2\alpha}. \label{dispersion}
\end{equation}

Given $\beta \in \mathbb{R}$, we recall the binomial expansion
\begin{equation}
(1 + x)^\beta = \sum_{j=0}^{\infty} \binom{\beta}{j} x^j,
\end{equation}
where $\binom{\beta}{j}$ denotes the generalized binomial coefficient, i.e.
\[
\binom{\beta}{j} = 
\begin{cases}
\dfrac{\beta (\beta - 1) \cdots (\beta - j + 1)}{j!} & \text{if } j \in \mathbb{N}, \\\\
1 & \text{if } j = 0.
\end{cases}
\]

Using this expansion, we obtain a power series representation of the energy:
\begin{equation}
\epsilon = \sum_{j=1}^{\infty} c_j \alpha^{j-1} \gamma^{2j},
\end{equation}
where the coefficients are defined as $c_j = \displaystyle{\binom{1/2}{j} 2^{2j - 1}}$.

By associating $\epsilon$ with the quantum mechanical operator corresponding to the momentum in coordinate representation,
\[
p = \hbar k \quad \longrightarrow \quad P = -i\hbar \nabla,
\]
we obtain the operatorial form of the energy:
\begin{equation}
\epsilon \rightarrow \sum_{j=1}^{\infty} c_j \left(-\frac{\hbar^2}{2m^*} \Delta\right)^j \alpha^{j-1}.
\end{equation}

This leads to the following generalized Schrödinger equation (GSE):
\begin{equation}
i\hbar \frac{\partial \Psi}{\partial t}(x,t) = \sum_{j=1}^{\infty} c_j \left(-\frac{\hbar^2}{2m^*}\right)^j \alpha^{j-1} \Delta^j \Psi(x,t) - qV(x)\Psi(x,t). \label{GSE}
\end{equation}

In practical applications, only a finite number of terms in the expansion are retained. The fourth-order equation represents the minimal extension beyond the effective mass approximation, and its analysis has been carried out in \cite{AliNaRo}. A study of higher-order dispersion Schr\"odinger equations can also be found in \cite{KaSha}.

Let $2s$, with $s \in \mathbb{N}$, denote the order of the expansion. The corresponding Hamiltonian is given by
\begin{equation}
\mathcal{H}_{2s} = \sum_{j=1}^{s} c_j \alpha^{j-1} \left( -\frac{\hbar^2}{2m^*} \right)^j \Delta^j - qV(x). \label{H_2s}
\end{equation}

In the cases considered in the following sections, the potential $V(x)$ is typically a real, piecewise regular function, as in the case of a resonant tunneling diode (RTD). This allows us to establish the following property.

\begin{proposition}
If $V$ is real and $V \in L^{\infty}(\mathbb{R}^3)$, then for any $s \in \mathbb{N}$, the Hamiltonian $\mathcal{H}_{2s}$ is well-defined on the Schwartz space $\mathcal{S}(\mathbb{R}^3)$ as a symmetric operator with respect to the scalar product in $L^2(\mathbb{R}^3)$.
\end{proposition}

The density of $\mathcal{S}(\mathbb{R}^3)$ in $L^2(\mathbb{R}^3)$ allows us to prove the following.

\begin{proposition} \label{prop:self_ad}
If $V$ is real and $V \in L^{\infty}(\mathbb{R}^3)$, then for any $s \in \mathbb{N}$, the Hamiltonian $\mathcal{H}_{2s}$ admits a self-adjoint extension in $H^{2s}(\mathbb{R}^3)$.
\end{proposition}

Here, $H^{2s}(\mathbb{R}^3)$ denotes the Sobolev space of complex-valued functions defined on $\mathbb{R}^3$ that belong to $L^2(\mathbb{R}^3)$ along with their (generalized) derivatives up to order $2s$.

We prove Proposition \ref{prop:self_ad} in several steps.

\begin{lemma}
If $V$ is real and $V \in L^{\infty}(\mathbb{R}^3)$, then for any $s \in \mathbb{N}$, the Hamiltonian $\mathcal{H}_{2s}$ is well-defined on $H^{2s}(\mathbb{R}^3)$ as a symmetric operator with respect to the scalar product in $L^2(\mathbb{R}^3)$.
\end{lemma}

\begin{proof}
Let $(\cdot,\cdot)_{L^2(\mathbb{R}^3)}$ denote the scalar product in $L^2(\mathbb{R}^3)$, assumed to be antilinear in the first argument. For any $f, g \in H^{2s}(\mathbb{R}^3)$ and $j \in \mathbb{N}$, $j \le 2s$, we have

\begin{eqnarray}
(\Delta^{j} f, g)_{ L^2(\R^3)} = \int_{\R^3} g(x)  \Delta^{j} \overline{f(x)} \, d x = \int_{\R^3} g(x) \nabla \cdot \nabla( \Delta^{j-1} \overline{f(x)}) \, d x = \nonumber\\
 - \int_{\R^3} \nabla  g(x) \cdot \nabla( \Delta^{j-1} \overline{f(x)}) \, d x =   \int_{\R^3} \Delta  g(x)  \Delta^{j-1} \overline{f(x)} \, d x = ... \nonumber \\
 = \int_{\R^3}  \overline{f(x)} \Delta^{j}   g(x)  \, d x = ( f, \Delta^{j} g)_{ L^2(\R^3)}.
\end{eqnarray}

Moreover, since $V(x)$ is real,
\begin{align*}
(f, Vg)_{L^2(\mathbb{R}^3)} &= \int_{\mathbb{R}^3} \overline{f(x)} V(x) g(x) \, dx = \int_{\mathbb{R}^3} \overline{V(x)f(x)} g(x) \, dx = (Vf, g)_{L^2(\mathbb{R}^3)}.
\end{align*}
Thus, the operator is symmetric.
\end{proof}

\noindent \emph{Proof of Proposition \ref{prop:self_ad}.}
Consider the free Hamiltonian $\mathcal{H}_{2s,0}$ in three dimensions
\[
\mathcal{H}_{2s,0} = \sum_{j=1}^{s} c_j \alpha^{j-1} \left( -\frac{\hbar^2}{2m^*} \right)^j \Delta^j.
\]
For $f \in H^{2s}(\mathbb{R}^3)$, we have
\[
\mathcal{H}_{2s,0} f = \frac{1}{(2\pi)^{3/2}} \int_{\mathbb{R}^3} e^{i k \cdot x} \left[ \sum_{j=1}^{s} c_j \alpha^{j-1} \left( \frac{\hbar^2}{2m^*} \right)^j |k|^{2j} \right] \hat{f}(k) \, dk.
\]
Define
\[
\mu_f^{\pm}(x) = \frac{1}{(2\pi)^{3/2}} \int_{\mathbb{R}^3} e^{i k \cdot x} \frac{\hat{f}(k)}{\pm i + \sum_{j=1}^{s} c_j \alpha^{j-1} \left( \frac{\hbar^2}{2m^*} \right)^j |k|^{2j}} \, dk,
\]
which belongs to $H^{2s}(\mathbb{R}^3)$. Then,

$$
({\cal H}_{2s,0} \pm i ) \mu_f^{\pm}  (x) = \frac{1}{(2\pi)^{3/2}} \int_{\R^3} e^{i k\cdot x} \hat{f}(k) \, d k = f(x) 
$$
implying that $\text{Range}(\mathcal{H}_{2s,0} \pm i) = L^2(\mathbb{R}^3)$, and hence $\mathcal{H}_{2s,0}$ is self-adjoint.

Since the multiplication operator $f \mapsto V f$ is bounded in $L^2(\mathbb{R}^3)$ by $\|V\|_{L^{\infty}(\mathbb{R}^3)}$ and symmetric, the full Hamiltonian $\mathcal{H}_{2s} = \mathcal{H}_{2s,0} - qV$ is also self-adjoint by the Kato-Rellich  theorem.

\hfill $\Box$

\begin{remark}
An alternative approach is to describe the system in terms of momentum or wave vector.
\end{remark}

Let $\hat{\Psi}(k,t)$ denote the wave function in wave vector space. It is related to $\Psi(x,t)$ via the unitary Fourier transform
\[
\hat{\Psi}(k,t) = \frac{1}{(2\pi)^{d/2}} \int_{\mathbb{R}^d} \Psi(x,t) e^{-i k \cdot x} \, dx, \quad d = 1,2,3,
\]
where the integral is understood in the principal value sense since we work in $L^2(\mathbb{R}^d)$. The function $\hat{\Psi}(k,t)$ satisfies the Schr\"odinger equation
\begin{align}
i\hbar \frac{\partial \hat{\Psi}(k,t)}{\partial t} = \epsilon(k) \hat{\Psi}(k,t) - \frac{q}{(2\pi)^d} \int_{\mathbb{R}^d \times \mathbb{R}^d} V(x) \hat{\Psi}(k',t) e^{i(k'-k)\cdot x} \, dk' \, dx. \label{GSE_k}
\end{align}

Introducing the Fourier transform $\hat{V}(k)$ of $V(x)$, the potential term becomes
\[
- \frac{q}{(2\pi)^{d/2}} \hat{V} * \hat{\Psi}(k,t),
\]
where $*$ denotes convolution.

Equation \eqref{GSE_k} allows for the inclusion of the full dispersion relation, but the potential term becomes integral in nature. However, for the simulation of electron devices such as RTDs, following \cite{LeKi}, we aim to solve the Schr\"odinger equation in a bounded domain representing the active region of the device, with appropriate boundary conditions. Formulating such boundary conditions in wave vector space is not straightforward; therefore, we focus exclusively on the position-space formulation given by Equation \eqref{GSE}.

\section{Generalization of the Current} \label{sec:curr}
We now want to calculate a generalized form for the current. We consider eq. (\ref{GSE}) and its conjugate one

\begin{align}
\label{SC}
&-i\hbar \frac{\partial \overline{\Psi}}{\partial t}=\sum_{j=1}^{+\infty}c_j\left(-\frac{\hbar^2}{2m^*}\right)^j\alpha^{j-1}\Delta^j \overline{\Psi}-qV \overline{\Psi}.
\end{align}

\noindent
If we multiply eq. (\ref{GSE}) for $\overline{\Psi}$ and eq. (\ref{SC}) for $\Psi$ and  subtract term by term, we get

\begin{equation}
i\hbar\frac{\partial}{\partial t}|\Psi|^2=\sum_{j=1}^{+\infty}c_j\left(-\frac{\hbar^2}{2m^*}\right)^j\alpha^{j-1}\left[\overline{\Psi}\Delta^j\Psi-\Psi\Delta^j\overline{\Psi}\right].
\label{generalizedS}
\end{equation}
Our aim is to put the previous equation in a divergence form.

\noindent

We observe that
\begin{align}
\overline{\Psi}\Delta^j\Psi 
&=\nabla \cdot\left( \overline{\Psi}\nabla(\Delta^{j-1}\Psi)\right)-\nabla \overline{\Psi}\cdot \nabla (\Delta^{j-1}\Psi)
\end{align}
\noindent
and similarly
\begin{equation}
\Psi\Delta^j\overline{\Psi}=\nabla \cdot\left( \Psi\nabla(\Delta^{j-1}\overline{\Psi})\right)-\nabla \Psi\cdot \nabla (\Delta^{j-1}\overline{\Psi}).
\end{equation}

\noindent
If $j=1$ we get the usual expression
\begin{equation}
 \overline{\Psi} \Delta  \Psi  - \Psi\Delta \overline{\Psi}   = \nabla \cdot\left(  \overline{\Psi} \nabla  \Psi  - \Psi\nabla \overline{\Psi} \right) = 2 i \nabla \cdot  \mathfrak{Im}
 ( \overline{\Psi} \nabla  \Psi ).
\end{equation}
 
 \noindent
Since for functions $f$ and $g$ regular enough it holds

$$\nabla \cdot \left((\nabla f )\Delta g\right)= \nabla f\cdot\nabla(\Delta g)+\Delta f\Delta g,$$

\noindent
if we set $f=\overline{\Psi}$ and $g=\Delta^{j-2}\Psi$, for $j > 2$, one has
\begin{equation}
\nabla \overline{\Psi}\cdot \nabla (\Delta^{j-1}\Psi)=\nabla\cdot \left(\nabla \overline{\Psi} \,\, \Delta^{j-1}\Psi \right)-\Delta^{j-1}\Psi\Delta \overline{\Psi}.
\end{equation}

\noindent
It follows that for $j \ge 2$
\begin{align}
&\overline{\Psi}\Delta^j\Psi-\Psi\Delta^j\overline{\Psi}=\nabla \cdot \Big( \overline{\Psi}\nabla(\Delta^{j-1}\Psi)-\Psi\nabla(\Delta^{j-1}\overline{\Psi})+(\nabla \Psi)\Delta^{j-1}\overline{\Psi}-(\nabla \overline{\Psi})\Delta^{j-1}\Psi\Big)+ \nonumber \\
&-\Delta^{j-1}\overline{\Psi}\Delta\Psi+\Delta^{j-1}\Psi\Delta\overline{\Psi}.
\end{align}

\noindent
If $j =2$ we have the sought divergence form
\begin{eqnarray}
\overline{\Psi}\Delta^2\Psi-\Psi\Delta^2\overline{\Psi}=
\nabla \cdot \left( \overline{\Psi}\nabla(\Delta \Psi)-\Psi\nabla(\Delta \overline{\Psi}) + 
(\nabla \Psi)\Delta \overline{\Psi}-(\nabla \overline{\Psi})\Delta \Psi  \right)  = \nonumber\\
2 i \nabla \cdot \mathfrak{Im} \left( \overline{\Psi}\nabla(\Delta \Psi) - (\nabla \overline{\Psi}) \Delta  \Psi  \right).
\end{eqnarray}
Let us suppose $j >2$. 
By iterating until in the term 
$$
 \nabla (\Delta^h \Psi ) \cdot \nabla  (\Delta^{r}\overline{\Psi})
$$
we get $h = r$, we have
\begin{eqnarray*}
&&\Delta \Psi\Delta^{j-1}\overline{\Psi}=\Delta \Psi \Delta (\Delta^{j-2}\overline{\Psi})= \Delta \Psi \nabla \cdot \nabla  (\Delta^{j-2}\overline{\Psi}) = 
\nonumber\\
&&\nabla  \cdot \left(\Delta \Psi \nabla (\Delta^{j-2}\overline{\Psi}) \right)  - \nabla (\Delta \Psi ) \cdot \nabla  (\Delta^{j-2}\overline{\Psi}) = \nonumber\\
&&\nabla  \cdot \left(\Delta \Psi \nabla (\Delta^{j-2}\overline{\Psi})  - \nabla  (\Delta \Psi) (\Delta^{j-2}\overline{\Psi}) \right) -  \Delta^2 \Psi \Delta^{j-2}\overline{\Psi} = \nonumber\\
&&\nabla  \cdot \left(\Delta \Psi \nabla (\Delta^{j-2}\overline{\Psi})  - \nabla  (\Delta \Psi) (\Delta^{j-2}\overline{\Psi})  -  (\Delta^2 \Psi) \nabla (\Delta^{j-3}\overline{\Psi})\right)  -  \nabla (\Delta^2 \Psi ) \cdot \nabla  (\Delta^{j-3}\overline{\Psi}).
\end{eqnarray*}

If we intend, for non negative integer $r$, the products of the nabla operator as follows
\begin{equation}
\nabla^r := \left\{
\begin{array}{ll}
\Delta^{r/2} & \mbox{if } r \quad \mbox{even}\\
\nabla \Delta ^{\frac{r-1}{2}} & \mbox{if } r \quad \mbox{odd}
\end{array}
\right. \label{nabla}
\end{equation}

\noindent
from the above relations one gets

\begin{align*}
\overline{\Psi}\Delta^j\Psi-\Psi\Delta^j\overline{\Psi} = 2i  \nabla\cdot  \mathfrak{Im} \left( \sum_{r=0}^{j-1}(-1)^{r} \nabla^{r}\overline{\Psi} \nabla^{2j -1 - r}\Psi \right).
\end{align*} 


Altogether, we have established the following result.

\begin{proposition}
Solutions to the generalized Schrödinger equation (GSE) satisfy the continuity equation
\[
\frac{\partial |\Psi|^2}{\partial t} + \nabla \cdot \mathbf{J} = 0,
\]
where the probability current density $\mathbf{J}$ is given by
\begin{equation}
\mathbf{J} = -\frac{2}{\hbar} \sum_{j=1}^{\infty} c_j \left(-\frac{\hbar^2}{2m^*}\right)^j \alpha^{j-1}
\, \mathfrak{Im} \left( \sum_{r=0}^{j-1} (-1)^r \nabla^r \overline{\Psi} \nabla^{2j - 1 - r} \Psi \right),
\end{equation}
and the powers of the nabla operator are interpreted as in equation~(\ref{nabla}).
\end{proposition}

In particular, for the fourth-order case, the current density becomes
\begin{equation}
\mathbf{J}_4 = \mathfrak{Im} \left( \frac{\hbar}{m^*} \overline{\Psi} \nabla \Psi + \frac{\alpha \hbar^3}{2 (m^*)^2} \left( \overline{\Psi} \nabla \Delta \Psi - \nabla \overline{\Psi} \Delta \Psi \right) \right).
\label{J4}
\end{equation}

For the sixth-order case, we obtain
\begin{equation}
\mathbf{J}_6 = \mathbf{J}_4 + \frac{\alpha^2 \hbar^5}{2 (m^*)^3} \mathfrak{Im} \left( \overline{\Psi} \nabla \Delta^2 \Psi - \nabla \overline{\Psi} \Delta^2 \Psi + \Delta \overline{\Psi} \nabla \Delta \Psi \right).
\label{J6}
\end{equation}

\begin{remark}
Consider the approximation of order $2s$, with $s \in \mathbb{N}$, and assume—according to spectral theory—that the solution belongs to the Sobolev space $H^{2s}(\mathbb{R}^d)$, where $d = 1, 2, 3$. In dimensions $d = 1$ and $d = 2$, the wave function $\Psi$ possesses continuous derivatives up to order $2s - 1$, ensuring that the current density is a smooth function. In the case $d = 3$, however, only continuity up to order $2s - 2$ can be guaranteed, and smoothness of the current density cannot be assured unless higher regularity beyond $H^{2s}(\mathbb{R}^d)$ is imposed.
\end{remark}

We note that a similar approach has been adopted for the study of spin current in \cite{BoDroFiWe}.

\section{Transparent boundary conditions for the generalized stationary Schr\"odinger equation} \label{sec:TBC}

We consider the stationary form of the Generalized Schrödinger Equation (GSE), given by

 $H\Psi=E\Psi$,

where the Hamiltonian operator is defined as
\begin{equation}
H =  \sum_{j=1}^{+\infty}c_j\alpha^{j-1}\left(\left(-\frac{\hbar^2}{2m^*}\right)^j\Delta^j \right)- qV(x). 
\end{equation}
This leads to the equation
\begin{equation}
 \sum_{j=1}^{+\infty}c_j\alpha^{j-1}\left(\left(-\frac{\hbar^2}{2m^*}\right)^j\Delta^j\Psi\right)-(qV(x)+E)\Psi(x)=0,
\end{equation}

Our objective is to derive transparent boundary conditions that are valid for any fixed order of the expansion. The methodology adopted here generalizes the approach introduced in \cite{AliNaRo}.

From this point onward, we restrict our analysis to the one-dimensional case. The real axis is partitioned into three regions: Region I ($x < 0$) and Region III ($x > L$, with $L > 0$) represent semi-infinite waveguides modeling the contacts, while Region II ($0 \le x \le L$) constitutes the active zone of the device. Regions I and III act as reservoirs from which electrons are injected into Region II. Electrons are injected from Region I with positive momentum and from Region III with negative momentum.

Let $2s$, with $s \in \mathbb{N}$, denote the order of the expansion in the Hamiltonian. We distinguish two cases based on the sign of the incident electron’s momentum. \\[0.5cm]

\subsection {Case $k_1>0$} 

Electron waves are injected at $x = 0$, and may be either reflected at $x = 0$ or transmitted at $x = L$. Generalizing the framework of \cite{AliNaRo}, we propose the following ansatz: the solution is given by the superposition of one incident wave and $s$ reflected ones in the region $x<0$ and  by the superimposition of $s$ transmitted waves in the region $x >L$, that is
\begin{equation}
\left\{
\begin{array}{ll}
\displaystyle \Psi_{I}(x)=e^{ik_1x}+\sum_{j=1}^{s} r_je^{-ik_jx}, &  x<0\\
\Psi_{III}(x)=\sum_{j=1}^{s} t_je^{i\tilde{k}_jx}, &  x>L
\end{array}
\right.  \label{ansatz1}
\end{equation}
where $r_j$ and $t_j$ represent the reflection and transmission coefficients to be determined along with the wave vectors $k_j$. 
\noindent

Assuming that $\Psi \in H^{2s}(\mathbb{R})$, we impose continuity of $\Psi$ and its derivatives up to order $2s - 1$ at $x = 0$, getting
\begin{align}
(ik_1)^l+\sum_{j=1}^{s}r_j(-ik_j)^l=\Psi_{II}^{(l)}(0)
\end{align}
\noindent
with $l=0,1,\ldots,2s-1$. 

If the wave vectors $k_i$ are distinct, i.e., $k_i \ne k_j$ for $i \ne j$, the first $s$ equations yield
\begin{equation}
r_j=\frac{\sum_{n=0}^{s-1}(-1)^{j+n+1}\left(\Psi_{II}^{(n)}(0)-(ik_1)^n \right)D_{n+1,j}^{x_j}}{\prod_{1\leq e < f \leq s}i(k_e-k_f)},\ j=1,\ldots,s
\end{equation}

\noindent
where $D_{n+1,j}^{x_j}$ is the cofactor of the $(n+1,j)$-entry of the matrix $D_{x_j}$, defined as
\[
D_{x_j} =
\begin{array}{c}
\begin{array}{cccccc}
      &  &  & & & \overset{\text{j}}{\downarrow} 
\end{array} \\[1ex]
\left[
\begin{array}{cccccc}
1 & 1 & \ldots& d_0 & \ldots &1 \\
-ik_1 & -ik_2 & \ldots & d_1 & \ldots & -ik_s\\
-k_1^2 & -k_2^2 & \ldots& d_2 & \ldots & -k_s^2 \\
\vdots & \vdots & \vdots & \vdots & \vdots & \vdots\\
(-ik_1)^{s-1} & (-ik_2)^{s-1} & \ldots & d_{s-1}& \ldots & (-ik_s)^{s-1}
\end{array}
\right].
\end{array}
\]

\noindent
Similarly, imposing continuity at $x = L$ gives

Substituting the expressions for $r_j$ into the remaining $s$ equations yields
\begin{equation}
(ik_1)^l+\sum_{j=1}^{s}\left[ \frac{\sum_{n=0}^{s-1}(-1)^{j+n+1}\left(\Psi_{II}^{(n)}(0)-(ik_1)^n \right)D_{n+1,j}^{x_j}}{\prod_{1\leq e < f \leq s}i(k_e-k_f)} \right](-ik_j)^l=\Psi_{II}^{(l)}(0)
\end{equation}

\noindent
for $l=s,\ldots,2s-1$.

\noindent
If we now  impose the continuity of $\Psi$ and all the derivatives up to the order $2s-1$ in $x=L$, we get $2s$ equations in $2s$ unknowns

\begin{align}
\sum_{j=1}^{s}t_j(i\tilde{k}_j)^{l-1}e^{i\tilde{k}_jL}=\Psi_{II}^{(l-1)}(L),\ l=1,\ldots,2s.
\end{align}

From the first $s$ equations, we get

\begin{equation}
t_j=\frac{\sum_{n=0}^{s-1}(-1)^{j+n+1}\Psi_{II}^{(n)}(L)E_{n+1,j}^{x_j}}{e^{iL\sum_{j=1}^s \tilde{k}_j}\prod_{1\leq e < f \leq s}i(\tilde{k}_f-\tilde{k}_e)},\ j=1,\ldots,s
\end{equation}

\noindent
where $E_{n+1,j}^{x_j}$ is the cofactor of the $(n+1,j)$-entry  of the matrix $E_{x_j}$
\[
E_{x_j} =
\begin{array}{c}
\begin{array}{cccccc}
      &  &  & & & \overset{\text{j}}{\downarrow} 
\end{array} \\[1ex]
\left[
\begin{array}{cccccc}
e^{i\tilde{k}_1 L} & e^{i\tilde{k}_2 L} & \ldots& \Psi_{II}(L) & \ldots & e^{i\tilde{k}_s L} \\
i\tilde{k}_1 e^{i\tilde{k}_1 L}  & i\tilde{k}_2 e^{i\tilde{k}_2 L}  & \ldots & \Psi_{II}^{'}(L)  & \ldots & i\tilde{k}_s e^{i\tilde{k}_s L} \\
-\tilde{k}_1^2 e^{i\tilde{k}_1 L}  & -\tilde{k}_2^2e^{i\tilde{k}_2 L}& \ldots&\Psi_{II}^{''}(L) & \ldots & -\tilde{k}_s^2e^{i\tilde{k}_s L}\\
\vdots & \vdots & \vdots & \vdots & \vdots & \vdots\\
(i\tilde{k}_1)^{s-1}e^{i\tilde{k}_1 L} & (i\tilde{k}_2)^{s-1}e^{i\tilde{k}_2 L} & \ldots & \Psi_{II}^{s-1}(L)& \ldots & (i\tilde{k}_s)^{s-1}e^{i\tilde{k}_s L}
\end{array}
\right]
\end{array}
\]

\noindent
If we substitute these $t_j$ in the remaining equations, one has the following boundary conditions

\begin{equation}
\sum_{j=1}^{s} \left[ \frac{\sum_{n=0}^{s-1}(-1)^{j+n+1}\Psi_{II}^{(n)}(L)E_{n+1,j}^{x_j}}{e^{iL\sum_{j=1}^s\tilde{k}_j}\prod_{1\leq e < f \leq s}i(\tilde{k}_f-\tilde{k}_e)} \right] (i\tilde{k}_j)^l e^{i\tilde{k}_j L}=\Psi^{(l)}(L), \ l=s,\ldots,2s-1
\end{equation}
If we have the coincidence of some wave vectors, the analysis is much more involved and the situation must be handled on a case-by-case basis.

Now we pass to evaluate the wave vectors  $k_i$. In the region $x<0$, after the substitution  the ansatz in the Schr\"odinger equation and exploiting the independence of the functions $e^{-ik_jx}, \ j=1,\ldots,s$ we get for the wave vectors ${k}_j$ the relation

\begin{equation}
\sum_{p=1}^s c_p\alpha^{p-1}\left(-\frac{\hbar^2}{2m^*} \right)^p(-ik_j)^{2p}-(qV(0)+E)=0, \ j=1,\ldots,s.
\end{equation}

\noindent 
By using the fact that $k_1$ is known, from the previous relation for $j = 1$ we evaluate the energy $E$
\begin{equation}
E=\sum_{p=1}^s c_p\alpha^{p-1}\left(-\frac{\hbar^2}{2m^*} \right)^p(ik_1)^{2p}-qV(0)
\end{equation}
\noindent
while the remaining relations allow us to get the wave vectors $k_j$, $ j = 2, \ldots,s$.
 
On the other hand, considering the region $x>L$, we get for the wave vectors $\tilde{k}_j$ the relation
\begin{equation}
\sum_{p=1}^s c_p\alpha^{p-1}\left(-\frac{\hbar^2}{2m^*} \right)^p(i\tilde{k}_j)^{2p}-(qV(L)+E)=0, \ j=1,\ldots,s.
\end{equation}
\begin{remark}
Only wave vectors with positive real parts are included in the ansatz. If any wave vector is complex,  to ensure bounded solutions, we take only the wave vectors having positive coefficient of the imaginary part and  evanescent modes arise.
\end{remark}

\subsection{Case $k_1<0$}

Electron waves are injected at the boundary point $x=L$, where they may either be reflected back into Region III or transmitted into Region I. The wave function is modelled by the following ansatz
\begin{equation}
\left\{
\begin{array}{ll}
\displaystyle \Psi_{I}(x)=\sum_{j=1}^{s}t_je^{-i\tilde{k}_j (x-L)}, &  x<0\\
\Psi_{III}(x)=e^{ik_1(x-L)}+\sum_{j=1}^{s}r_je^{-ik_j(x-L)}, &  x>L
\end{array}
\right. \label{ansatz2}
\end{equation}

\noindent
To ensure physical consistency, we impose the continuity of the wave function $\Psi$ and its derivatives up to order $2s-1$  at the interface $x=0$. This yields the following system of equations:
\begin{align}
\sum_{j=1}^{s} t_j(-i\tilde{k}_1)^l e^{i\tilde{k}_j L}=\Psi_{II}^{(l)}(0), \quad l = 0, 1, \ldots, 2s - 1.
\end{align}

\noindent
Assuming distinct wave vectors $\tilde{k}_i \ne \tilde{k}_j$ for $i \ne j$, the first $s$ equations allow us to solve for the transmission coefficients $ t_j $:

\begin{equation}
t_j=\frac{\sum_{n=0}^{s-1}(-1)^{j+n+1}\Psi_{II}^{(n)}(0)E_{n+1,j}^{x_j}}{e^{iL\sum_{j=1}^{s}\tilde{k}_j}\prod_{1\leq e < f \leq s}i(\tilde{k}_e-\tilde{k}_f)},\ j=1,\ldots,s
\end{equation}

\noindent
where $E_{n+1,j}^{x_j}$ denotes the cofactor of the $(n+1, j)$-entry of the matrix $E_{x_j}$
\[
E_{x_j} =
\begin{array}{c}
\begin{array}{cccccc}
      &  &  & & & \overset{\text{j}}{\downarrow} 
\end{array} \\[1ex]
\left[
\begin{array}{cccccc}
e^{i\tilde{k}_1 L} & e^{i\tilde{k}_2 L} & \ldots& \Psi_{II}(L) & \ldots & e^{i\tilde{k}_s L} \\
-i\tilde{k}_1 e^{i\tilde{k}_1 L}  &- i\tilde{k}_2 e^{i\tilde{k}_2 L}  & \ldots & \Psi_{II}^{'}(L)  & \ldots & -i\tilde{k_s}e^{i\tilde{k}_s L} \\
-\tilde{k}_1^2e^{i\tilde{k}_1 L}  & -\tilde{k}_2^2e^{i\tilde{k}_2 L}& \ldots&\Psi_{II}^{''}(L) & \ldots & -\tilde{k}_s^2e^{i\tilde{k}_s L}\\
\vdots & \vdots & \vdots & \vdots & \vdots & \vdots\\
(-i\tilde{k}_1)^{s-1}e^{i\tilde{k}_1 L} & (-i\tilde{k}_2)^{s-1}e^{i\tilde{k}_2 L} & \ldots & \Psi_{II}^{s-1}(L)& \ldots & (-i\tilde{k}_s)^{s-1}e^{i\tilde{k}_s L}
\end{array}
\right]
\end{array}
\]

\noindent
Substituting the expressions for $ t_j $ into the remaining equations yields the transparent boundary conditions at $x = 0$
\begin{equation}
\sum_{j=1}^s \left[ \frac{\sum_{n=0}^{s-1}(-1)^{j+n+1}\Psi_{II}^{(n)}(0)E_{n+1,j}^{x_j}}{e^{iL\sum_{j=1}^{s}\tilde{k}_j}\prod_{1\leq e < f \leq s}i(\tilde{k}_e-\tilde{k}_f)}\right](-i\tilde{k}_j)^l e^{i\tilde{k}_j L}=\Psi^{(l)}(0), \ l=s,\ldots,2s-1 \label{TBC0}
\end{equation}

\noindent
\noindent
Next, we impose continuity of $\Psi$ and its derivatives up to order $2s - 1$ at $x = L$, resulting in
\begin{equation}
(ik_1)^{l} +\sum_{j=1}^s r_j(-ik_j)^{l} =\Psi^{(l)}(L), \quad l = 0, 1, \ldots, 2s-1.
\end{equation}

\noindent
Following the same procedure, we obtain the reflection coefficients $r_j$

\begin{equation}
r_j=\frac{\sum_{n=0}^{s-1}(-1)^{j+n+1}\left(\Psi_{II}^{(n)}(L)-(ik_1)^n \right)D_{n+1,j}^{x_j}}{\prod_{1\leq e < f \leq s}i(k_e-k_f)},\ j=1,\ldots,s
\end{equation}

\noindent
and the corresponding boundary conditions
\begin{equation}
(ik_1)^l +\sum_{j=1}^{s}\left[ \frac{\sum_{n=0}^{s-1}(-1)^{j+n+1}\left(\Psi_{II}^{(n)}(L)-(ik_1)^n \right)D_{n+1,j}^{x_j}}{\prod_{1\leq e < f \leq s}i(k_e-k_f)} \right](-ik_j)^l=\Psi_{II}^{(l)}(L). \label{TBCL}
\end{equation}
\noindent
for $l=s,\ldots,2s-1$.

To determine the wave vectors, we analyze the dispersion relations in the respective regions. For $ x < 0$, the wave vectors $\tilde{k}_j$ satisfy
\begin{equation}
\sum_{p=1}^s c_p\alpha^{p-1}\left(-\frac{\hbar^2}{2m^*} \right)^p(-i\tilde{k}_j)^{2p}-(qV(0)+E)=0, \ j=1,\ldots,s.
\end{equation}
Evaluating for $j=1$ we get the energy
\begin{equation}
E=\sum_{p=1}^s c_p\alpha^{p-1}\left(-\frac{\hbar^2}{2m^*} \right)^p(ik_1)^{2p}-qV(L).
\end{equation}
while the remaining relations gives the values of the remaining $k_j$'s, $j=2, \ldots,s$.
\noindent
For $ x > L$, the wave vectors $k_j$ satisfy
\begin{equation}
\sum_{p=1}^s c_p\alpha^{p-1}\left(-\frac{\hbar^2}{2m^*} \right)^p(-i k_j)^{2p}-(qV(L)+E)=0, \ j=1,\ldots,s.
\end{equation}

If we have the coincidence of some wave vectors, as for the case $k_1>0$ the analysis is much more involved and the situation must be handled on a case-by-case basis. 
\begin{remark}
Only wave vectors with positive real parts are included in the ansatz. If any wave vector is complex, to ensure bounded solutions, we take only the wave vectors having negative coefficient of the imaginary part and  evanescent modes arise.
\end{remark}

\section{Well posedness of the higher order SE with transparent boundary conditions}\label{sec:well}

We establish a well-posedness result for the generalized Schr\"odinger  model of arbitrary even order, thereby extending the result previously proven for the fourth-order case. For the sake of clarity, we present the case where $ k_1 > 0 $; the case $ k_1 < 0$ can be treated analogously.

Consider the generalized Schr\"odinger equation of order $2s$, $s\in\N$
\begin{equation}
\sum_{j=1}^{s}c_j \alpha^{j-1}\left(\left(-\frac{\hbar^2}{2m^*} \right)^j \Psi^{(2j)}(x)\right)-(qV(x)+E)\Psi=0.
\label{SEGeneral}
\end{equation}

Let $W^{2s,1}(0,L)$ denote the Sobolev space defined by
$$W^{2s,1}(0,L)=\{u\in L^1(0,L):D_\alpha u\in L^1(0,L),\ \forall \alpha : |\alpha| \le 2s \}.$$

\noindent
\begin{proposition}
Assume that the potential $V(x)$ is real-valued and belongs to $L^{\infty}(0,L)$. Then, the boundary value problem consisting of the generalized Schrödinger equation (\ref{SEGeneral}) of order \( 2s \), together with the transparent boundary conditions (\ref{TBC0})-(\ref{TBCL}), admits a unique solution  $\Psi \in W^{2s,1}(0,L)$, provided that $k_i \neq k_j$ for all $i \ne j$, $i,j = 0,1,2,3$, and the matrix \( A = (a_{hr}) \), whose entries are defined as:
\begin{eqnarray*}
a_{hr} = \left\{\begin{array}{ll}\sum_{j=1}^{s}(-1)^{j+r}D_{r}^{x_j}(-ik_j)^{h+s-1} & h,r = 1,2, \ldots, s\\[0.4cm]
 \prod_{1\leq e < f\leq s}i({k}_f-{k}_e) & h = 1,2, \ldots, s, \quad r = h+s,\\[0.4cm]
\sum_{j=1}^{s}\left(\sum_{n=0}^{s-1}(-1)^{j+n+1}\varphi_{r-1}^{(n)}(L)E_{n+1,j}^{x_j}\right)(i\tilde{k}_j)^{h-1}e^{i\tilde{k}_j(h-1)} & \\
- \varphi_{r-1}^{(l)}(L)\left[e^{iL\sum_{j=1}^{s}\tilde{k}_j}\prod_{1\leq e < f\leq s}i(\tilde{k}_f-\tilde{k}_e) \right] 
& h = s+1, \ldots, 2s, r = 1, 2, \ldots, 2s.\\[0.4cm]
0 & \mbox{otherwise}
\end{array} \right.
\end{eqnarray*}
is non singular.
\end{proposition}
\begin{proof} Under the hypothesis on the potential, the coefficients of the linear equations (\ref{SEGeneral}) are in $L^1_{loc} (0,L)$ ensuring the applicability of the existence and uniqueness theorem for the associated Cauchy problem. Let $(\varphi_0,\varphi_1,\varphi_2,\ldots,\varphi_{2s-1})$ be a fundamental set of solutions to \ref{SEGeneral}, satisfying the initial conditions:

\begin{align}
&\varphi_0(0)=1,\quad \varphi_0^{(1)}(0)=0,\ldots, \quad \varphi_0^{(2s-1)}(0)=0, \nonumber \\
&\varphi_1(0)=0,\quad \varphi_1^{(1)}(0)=1,\ldots, \quad \varphi_1^{(2s-1)}(0)=0, \nonumber \\
& \vdots \nonumber \\
&\varphi_{2s-1}(0)=0,\quad \varphi_{2s-1}^{(1)}(0)=0,\ldots, \quad \varphi_{2s-1}^{(2s-1)}(0)=1. \nonumber \\\nonumber
\end{align}

\noindent
The general integral of (\ref{SEGeneral}) can be expressed as  
$\Psi(x)=\sum_{p=0}^{2s-1}c_p\varphi_p(x)$ with $c_p \in \C$. 

Substituting this expression into the transparent boundary conditions at $x=0$ yields

\begin{align}
&\sum_{n=0}^{s-1}\left[\sum_{j=1}^{s}(-1)^{j+n+1}D_{n+1}^{x_j}(-ik_j)^l\right]c_n-\left[\prod_{1\leq e < f\leq s}i(k_e-k_f)\right]c_l=\nonumber \\
&\sum_{j=1}^{s}\sum_{n=0}^{s-1}(-1)^{j+n+1}(ik_1)^nD_{n+1}^{x_j}(-ik_j)^l-(ik_j)^l\left[\prod_{1\leq e < f\leq s}i(k_e-k_f)\right],
\end{align}
\noindent
for $l=s,\ldots,2s-1$. 

At $x = L$, the boundary condition becomes:
\begin{align}
&\sum_{p=0}^{2s-1}c_p \left\{\sum_{j=1}^{s}\left(\sum_{n=0}^{s-1}(-1)^{j+n+1}\varphi_p^{(n)}(L)E_{n+1,j}^{x_j}\right)(i\tilde{k}_j)^le^{i\tilde{k}_jl} \right.\nonumber \\
& - \left.\varphi_p^{(l)}(L)\left[e^{iL\sum_{j=1}^{s}\tilde{k}_j}\prod_{1\leq e < f\leq s}i(\tilde{k}_f-\tilde{k}_e) \right]\right\} = 0,
\end{align}
\nonumber 
$l=s,\ldots,2s-1$.

These conditions yield a linear system for the coefficients $c_p$
\begin{eqnarray}
A {\bf c}_p = {\bf b}
\end{eqnarray}
with obvious meaning of ${\bf b}$.
Therefore if $A$ is invertible we have the existence and uniqueness of the solution to the boundary problem defined by (\ref{SEGeneral}) and the transparent boundary conditions.
\end{proof}


\section{Comparison of the current between the second and the fourth order SE} \label{sec:comparison}

In this section, we investigate the influence of the additional terms introduced in higher-order Schrödinger equations (SEs) on the probability current. To minimize algebraic complexity, we focus on the comparison between the second-order (SE2) and fourth-order (SE4) formulations. For clarity, the analysis is restricted to the one-dimensional case.

We assume that all wave components are propagating, i.e., all wave vectors are real. For definiteness, we consider electrons incident from the left with a positive wave vector.

The probability current associated with SE2 is given by
 
$$J_{2}=\frac{\hbar}{m^*}  \mathfrak{Im} (\bar{\Psi}\Psi').$$

The incident wave is $\Psi_{inc}=e^{ikx_1}$ while the reflected and transmitted ones are 
$\Psi_{refl}=re^{-ikx_1}$ and $\Psi_{transm}=te^{i\tilde{k}_1 x}$.
The corresponding incident, reflected and transmitted  currents read
\begin{align}
&J_{inc}=\frac{\hbar}{m^*}  \mathfrak{Im} (e^{-ik_1 x}i k_1 e^{i k_1 x})=\frac{\hbar}{m^*}k_1, \\
&J_{refl}=\frac{\hbar}{m^*}  \mathfrak{Im}(\bar{r} e^{ik_1x}(-i k_1 r) e^{-i k_1 x})=-\frac{\hbar}{m^*} k_1|r|^2, \\
&J_{transm}=\frac{\hbar}{m^*}  \mathfrak{Im}(\bar{t} e^{-i \tilde{k}_1 x}(i t \tilde{k}_1)e^{i \tilde{k}_1 x})=\frac{\hbar}{m^*} \tilde{k}_1 |t|^2.
\end{align}
which give the following transmission and reflection probabilities
\begin{align}
 |T|^2:=\frac{|J_{transm}|}{|J_{inc}|}=\frac{|\tilde{k}_1|}{|k|}|t|^2,\quad
 |R|^2:=\frac{|J_{refl}|}{|J_{inc}|}=|r|^2
\end{align}
satisfying
$$
 |R|^2 + |T|^2 = 1
$$
because of the conservation of the current.

Defining the total wavefunction on the left as \(\Psi_{\text{left}} = e^{ik_1x} + r e^{-ik_1x}\), the total current in this region becomes
\begin{align}
J_{left}=&\frac{\hbar}{m^*}  \mathfrak{Im} ((e^{-i k_1 x}+\bar{r}e^{i k_1 x})(ike^{i k_1 x}-i r k_1 e^{-i k_1 x}))= \nonumber \\
=&\frac{\hbar}{m^*}  \mathfrak{Im} \left(ik_1+2Re(ik_1\bar{r} e^{2ik_1x})-ik_1 |r|^2  \right)=\frac{\hbar}{m^*}k_1-\frac{\hbar}{m^*}k_1|r|^2=\nonumber \\
=&J_{inc}+J_{refl}=\frac{\hbar}{m^*}k_1(1-|r|^2)=\frac{\hbar}{m^*}\tilde{k}_1 |t|^2=J_{trasm}.
\end{align}

%
Now let us consider SE4. In this case the current reads

$$J= J_4 = \frac{\hbar}{m}  \mathfrak{Im} (\Psi^{'}\bar{\Psi})+\frac{\alpha \hbar^3}{2m^{*2}}  \mathfrak{Im} (\Psi^{'''}\bar{\Psi}-\Psi^{''}\bar{\Psi}^{'})$$

The wavefunctions are $\Psi_{inc}=e^{ik_1x},\ \Psi_{refl}=r_1e^{-ik_1x}+r_2e^{-ik_2x},\  \Psi_{transm}=t_1e^{i \tilde{k}_1 x}+t_2e^{i \tilde{k}_2 x}$.
The corresponding currents are

\begin{align}
J_{inc}&= \frac{\hbar}{m^*}k_1-\frac{\alpha \hbar^3}{m^{*2}}k_1^3,\\
J_{refl}
&=- \frac{\hbar}{m^*}\Big[k_1|r_1|^2 + \mathfrak{Re}\Big(\bar{r_1}r_2k_2e^{i(k_1-k_2)x} + \bar{r_2}r_1k_1e^{i(k_2-k_1)x}\Big) + |r_2|^2k_2\Big]\nonumber \\
&+\frac{\alpha \hbar^3}{2m^{*2}}\Big[ 2 k_1^3|r_1|^2+ \mathfrak{Re} \Big( (k_1+k_2) \Big( k_2^2 \bar{r_1}r_2e^{i(k_1-k_2)x}+ k_1^2  \bar{r_2}r_1e^{i(k_2-k_1)x}\Big)\Big)  +  2 |r_2|^2k_2^3 \Big],\\ 
J_{transm}
&=\frac{\hbar}{m^*}\Big[\tilde{k}_1 |t_1|^2 + \mathfrak{Re}\Big(\bar{t_1}t_2  \tilde{k}_2 e^{i( \tilde{k}_2 -\tilde{k}_1)x}+ \bar{t_2}t_1\tilde{k}_1 e^{i(\tilde{k}_1 -  \tilde{k}_2)x}\Big)+|t_2|^2  \tilde{k}_2\Big]+\nonumber \\
&-\frac{\alpha \hbar^3}{2m^{*2}}\Big[2 \tilde{k}_1^3|t_1|^2+ \mathfrak{Re}\Big((\tilde{k}_1 + \tilde{k}_2) \Big(\tilde{k}_1^2 \bar{t_2}t_1 e^{i( \tilde{k}_1- \tilde{k}_2 )x}    
+  \tilde{k}_2^2 \bar{t_1}t_2 e^{i( \tilde{k}_2 - \tilde{k}_1 )x} \Big) \Big) + 2  \tilde{k}_2^3|t_2|^2\Big] \label{trans_curr}
\end{align}

\noindent
If we set $\Psi_{tot}=e^{ik_1x}+r_1e^{-ik_1x}+r_2e^{-ik_2x}$, the associated current reads  

\begin{align}
J_{tot}=&\frac{\hbar}{m^*} \mathfrak{Im} ((e^{-ik_1x}+\bar{r}_1e^{ik_1x}+\bar{r}_2e^{ik_2x})(ik_1e^{ik_1x}-ik_1r_1e^{-ik_1x}-ik_2r_2e^{-ik_2x}))+\nonumber \\
&+\frac{\alpha \hbar^3}{2m^{*2}} \mathfrak{Im} ((e^{-ik_1x}+\bar{r}_1e^{ik_1x}+\bar{r}_2e^{ik_2x})(-ik_1^3e^{ik_1x}+ik_1^3r_1e^{-ik_1x}+ik_2^3r_2e^{-ik_2x}))+\nonumber \\
&-\frac{\alpha \hbar^3}{2m^{*2}} \mathfrak{Im} ((-ik_1e^{-ik_1x}+ik_1\bar{r_1}e^{ik_1x}+ik_2\bar{r_2}e^{ik_2x})(-k_1^2e^{ik_1x}-k_1^2r_1e^{-ik_1x}-k_2^2r_2e^{-ik_2x}))=\nonumber \\
&=J_{inc}+J_{refl}+J_{extra},
\end{align}

\noindent
where 
\begin{align*}
J_{extra} 
&=\frac{\hbar}{m^*} \mathfrak{Im} \left[ ir_2 e^{-i(k_1+k_2)x}\left(-k_2+\frac{\alpha\hbar^2}{2m^*}\left(k_2^3-k_1k_2^2\right)\right) \right]\\
& +\frac{\hbar}{m^*} \mathfrak{Im} \left[ -i\bar{r}_2\overline{e^{-i(k_1+k_2)x}}\left(-k_1+\frac{\alpha\hbar^2}{2m^*}\left(k_1^3-k_1^2k_2\right)\right) \right].\\
\end{align*}

If we set
$$
r_2 e^{-i(k_1+k_2)x} = c+i d,
$$
after some algebra one has
\begin{align*}
J_{extra} & =\ \frac{\hbar}{m^*} c (k_1-k_2)\left[ 1-\frac{\alpha\hbar^2}{2m^*}\left( k_2^2+k_1^2 \right) \right].
\end{align*}
Since $k_1^2+k_2^2=\dfrac{2m^*}{\alpha\hbar^2}$ (see equation  \ref{equ}) we get $J_{extra}=0$ and therefore 
the conservation of the current 
\begin{equation}
J_{tot} = J_{inc} + J_{refl} = J_{transm}.
\end{equation}

In Fig.\ref{Jtrans} we have plotted the ensemble incident, reflected and transmitted probability current, obtained numerically, versus the wave vector of the incident wave. 

\begin{figure}[H]
    \centering
{
        \includegraphics[width=0.55\textwidth]{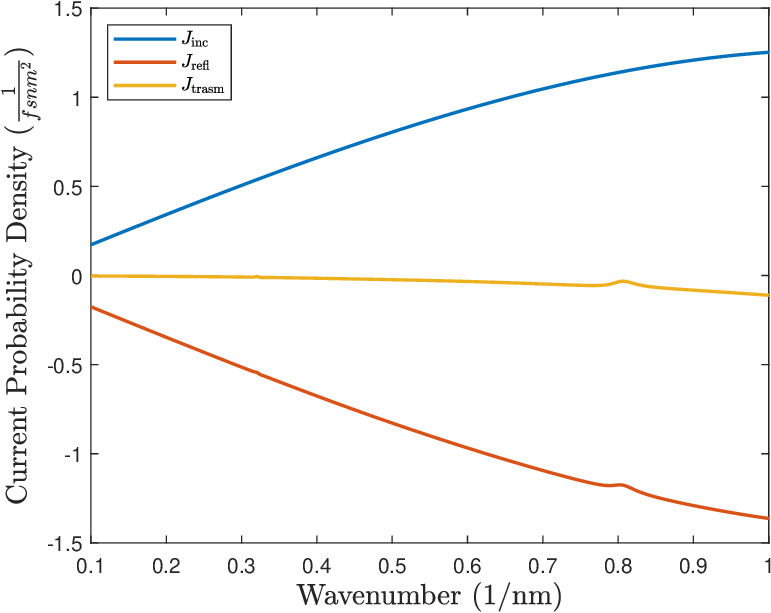}
       }
 \caption{Numerical evaluation of the ensemble incident, reflected and transmitted probability current density versus the wave vector of the incident wave.}
    \label{Jtrans}
\end{figure}
\noindent
\begin{remark} \label{rem:waves} Evanescent modes are absent when all wave vectors
$k_i\in \mathbb{R}$ for any $i = 1,2,3,4$. The wave vectors $k_1$ and $k_2$ satisfy the quartic equation
\begin{align}
ax^4+\frac{\hbar^2}{2m^*}x^2-(qV(x_b)+E)=(x^2-k_1^2)(ax^2+\frac{\hbar^2}{2m^*}+ak_1^2)=0
\label{equ}
\end{align}
with \begin{equation}
E=\frac{\hbar^2}{2m^*}k_1^2+ak_1^4-qV(x_b),
\end{equation}
$ a=-\frac{\hbar^4}{4(m^*)^2}\alpha <0$ and $x_b=0$ if the electron enters from the left contact, $x_b = L$ if the electron enters from the right contact.
\noindent
Equation \ref{equ} admits only real solutions if and only if $\frac{\hbar^2}{2m^*}+ak_1^2>0$, that is, if $ - \sqrt{\dfrac{2m^*}{\alpha \hbar^2}} < k_1<\sqrt{\dfrac{2m^*}{\alpha \hbar^2}}$. 

The other wave vectors are given by
$$
k_{3,4} = \sqrt{-\frac{\hbar^2}{4 m^* a} \pm \frac{1}{2 a} \sqrt{\frac{\hbar^4}{4 (m^*)^2} + 4 a (q V(x_c) + E)}}
$$
where in the complex case the root must be intended with positive real and imaginary parts. Here  
$x_c=0$ if the electron enters from the right contact, $x_b = L$ if the electron enters from the left  contact. Therefore a priori it is not possible to exclude that $k_3$ or $k_4$ lead to evanescent waves. 
\end{remark}


\section{Simulation of a resonant tunneling diode} \label{sec:RTD}
 
 We apply the theoretical framework developed in the previous sections to the case of a RTD, considering the electrostatic potential $V(x)$ proposed in \cite{HeWa} (see Fig.~\ref{fig:rampa}) and also numerically investigated in \cite{AliNaRo}:
 
$$
V(x) = \left\{
\begin{array}{ll}
    0 & \mbox{if } x \in [0, a_1[ \\[0.2cm]
   (x - a_1) \frac{V_L}{a_6 - a_1} &  \mbox{if } x \in [a_1, a_2]\cup \in [a_3, a_4]\cup [a_5, a_6] \\[0.2cm]
  (x - a_1) \frac{V_L}{a_6 - a_1} + V_b & \mbox{if } x \in ]a_2, a_3[ \cup ]a_4, a_5[\\[0.2cm]
V_L & \mbox{if } x \in ]a_6,L]
\end{array}
\right. 
$$
which corresponds to the superposition of a double barrier structure and a linear potential. In our simulations, we set the parameters as follows: $\alpha = $ 0.242 eV$^{-1}$, $a_1 = $ 50 nm, $a_2 =$ 60 nm, $a_3 = $ 65 nm, $a_4 =$ 70 nm, $a_5 = $ 75 nm, $a_6 =$ 85 nm, $L =$ 135 nm,  $V_b =-0.3$ V. 

\begin{figure}[H]
    \centering
{
        \includegraphics[width=0.6\textwidth]{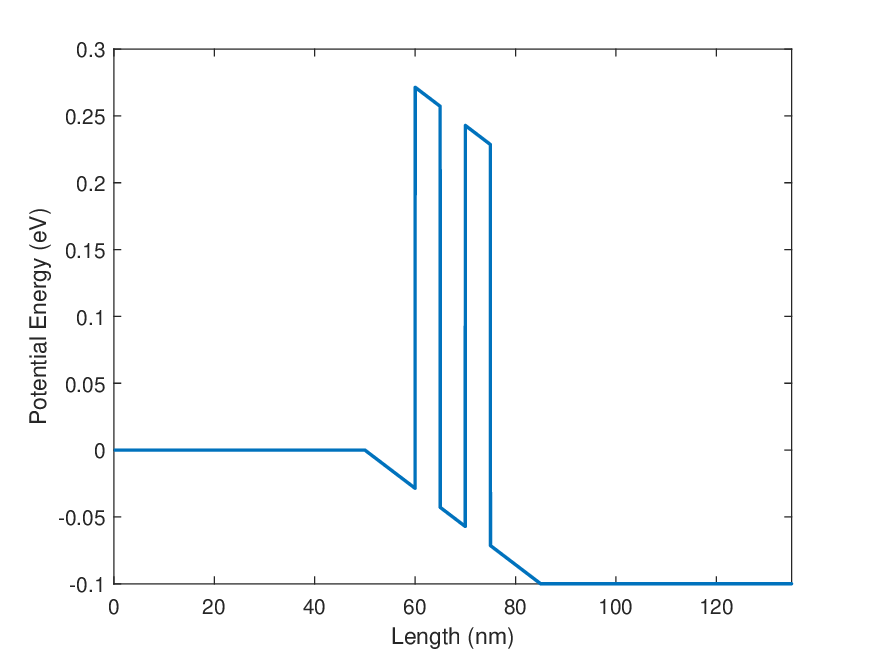}
       }
 \caption{RTD type potential energy in the case $V_0=0$ V, $V_L= 0.1$ V, $V_b = -0.3$ V,  $L=135$ nm. }
    \label{fig:rampa}
\end{figure}

\begin{figure}[H]
    \centering
{
        \includegraphics[width=0.6\textwidth]{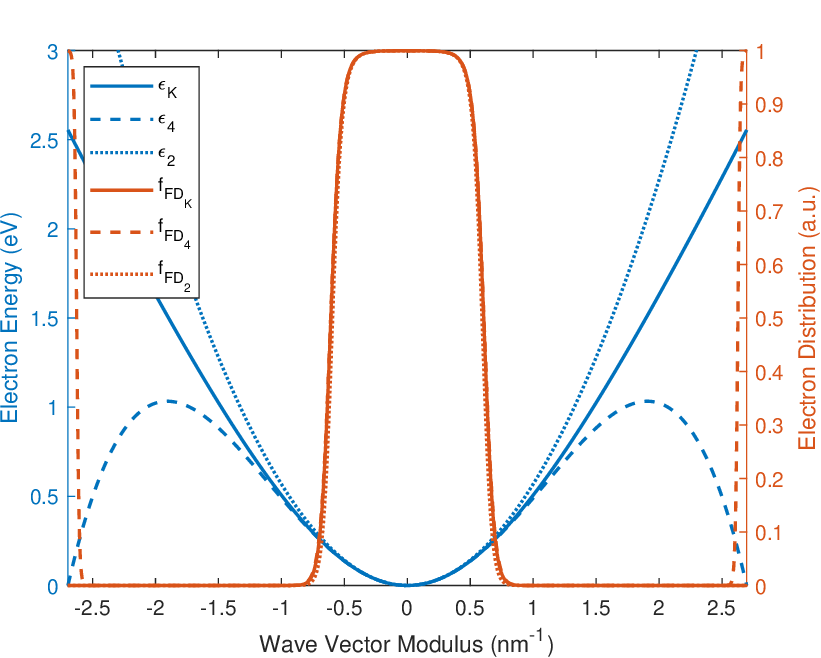}
       }
 \caption{Plot of the full Kane dispersion relation $\epsilon_K$ (continuous line), fourth order approximation $\epsilon_4$ (dashed line) and second order approximation $\epsilon_2$ (dotted line) along with the Fermi-Dirac distribution evaluated with the previous expressions of the energy band, denoted as $f_{FD_K}$, $f_{FD_4}$, $f_{FD_2}$, respectively.  
When the quartic approximation is adopted  spurious tails are present close to $\pm \sqrt{\frac{2 m^*}{\alpha \hbar^2}}$ but they are well outside from the region of physical interest. }
    \label{fig:distribution}
\end{figure}

To evaluate the key quantities relevant for RTD design - namely, the average electron density and current - we consider an incident plane wave with wave vector $k$. The associated density operator in coordinate representation has kernel $\rho_k(x,y) = \Psi_k(x) \overline{\Psi_k(y)}$, where $\Psi_k(x)$ is the solution of the corresponding Schr\"odinger equation.
The kernel of the ensemble density matrix operator is given by
\begin{equation}
\rho (x, y) = \frac{g_s g_v}{(2\pi )^3} \int_{\R^3} f_{FD}(k) \Psi_k (x) \overline{\Psi_k (y)} \, d^3 k,
\end{equation}
where $g_s$ and $g_v$ denote the spin and valley degeneracies, respectively, and $f_{FD}$ is the Fermi-Dirac distribution
$$f_{FD}(k)=\dfrac{1}{1+\exp\left(\frac{E(k)-E_F}{k_BT}\right)},$$
with $E(k)$ being the electron energy, $E_F$ the Fermi energy and $T$ the absolute temperature which we set 300 $K$ (room temperature). 
We observe that 
$$
E(k)-E_F = E(k) - E_C - (E_F - E_C) = \epsilon(k) - \mu
$$
with $\mu = E_F - E_C$ chemical energy which depends on the doping or the injected charge and not on the applied electric field. In the sequel we set 
$\mu = $ 0.2 eV which correspond to a doping of about 7.8$\times$10$^{18}$ cm$^{-3}$ in both contacts.

In principle the Kane approximation is considered valid for any wave vector in $\R^3$. In Fig.~\ref{fig:distribution} 
    we have plotted $\epsilon (k)$ along with $f_{FD}(k)$ versus  the modulus $k$ of the wave vector comparing the  full Kane expression with  the quartic  and quadratic (parabolic) approximation. One observes that from a numerical point of view  $f_{FD}(k)$ exhibits a compact support. Within such a set the full Kane dispersion relation is accurately approximated by the quartic expansion; indeed, also the quadratic approximation is acceptable,  although slightly less precise. Fig.~\ref{fig:distribution} further shows that the quartic approximation is no longer monotone for $k > \sqrt{\dfrac{m^*}{\alpha \hbar^2}}$. The loss of monotonicity is not  problematic since the interval $\left[- \sqrt{\dfrac{m^*}{\alpha \hbar^2}}, \sqrt{\dfrac{m^*}{\alpha \hbar^2}} \right] $ fully contains the support of $f_{FD}(k)$. Spurious tails of the distribution function with the quartic approximation appear  close to $\pm \sqrt{\dfrac{2 m^*}{\alpha \hbar^2}}$ but these lie outside the  domain of numerical integration. 
It is worth noting that, despite the  similarity of the resulting statistics, the parabolic and quartic dispersion relation leads to different Schr\"odinger equations. So one expects different results for the average quantities of interest.
   
The electron density $n(x)$ is obtained as
\begin{equation}
n(x) = \mbox{tr} \ \rho (x, y) = \frac{g_s g_v}{(2\pi)^3} \int_{\R^3} f_{FD}(k) |\Psi_k (x)|^2 \, d^3 k,
\end{equation}
where $\mbox{tr}$ denotes the trace.
\noindent
The total electric current density is given by
$${\cal J}= - q \frac{g_s g_v}{(2\pi )^3}\int_{\R^3} f_{FD}(k)J_{k_x} d^3 k.$$

By adopting cylindrical coordinates: $k_x = k_x$, $k_y = \sigma \cos\theta$, $k_z = \sigma \sin\theta$, with $\sigma \in [0, +\infty[$ and $\theta \in [0, 2\pi[$, the latter becomes
\begin{equation}
{\cal J} = -q \frac{g_s g_v}{(2\pi)^2}\int_{-\infty}^{+\infty}dk_x \int_{0}^{+\infty}\frac{J_{k_x} \sigma}{1+\exp\left( \dfrac{\epsilon(\sigma,k_x)- \mu}{k_BT}\right)}d\sigma, 
\end{equation}
where $\epsilon(\sigma,k_x)=\dfrac{-1+\sqrt{1+\dfrac{2\alpha\hbar^2(\sigma^2+k_x^2)}{m^{*}}}}{2\alpha}$.
\vskip 0.3cm
Figure~\ref{fig:integrand} shows the plot of the integrand function
$$
g(\sigma, k_x) = \sigma \left[1+\exp\left( \dfrac{\epsilon(\sigma,k_x)-\mu}{k_BT}\right)\right]^{-1}.
$$
which numerically exhibits a compact support consistently with the consideration outlined above.
Therefore we restrict the numerical domain to $- \overline{k} \le k_x \le \overline{k}$  
and $0\le \sigma \le  \overline{\sigma}$, 
with $\overline{k} = $1.5 nm$^{-1}$ and 
$\overline{\sigma} = $ 1 nm$^{-1}$. Within this range  the full Kane's dispersion relation and the quartic approximation are in good agreement as shown by Figure~\ref{fig:FD_log}  which in logarithmic scale highlights in detail the differences  between $f_{FD}(k)$ with the full Kane dispersion relation and with the quartic and quadratic approximation.

\begin{figure}[H]
    \centering
{
        \includegraphics[width=0.58\textwidth]{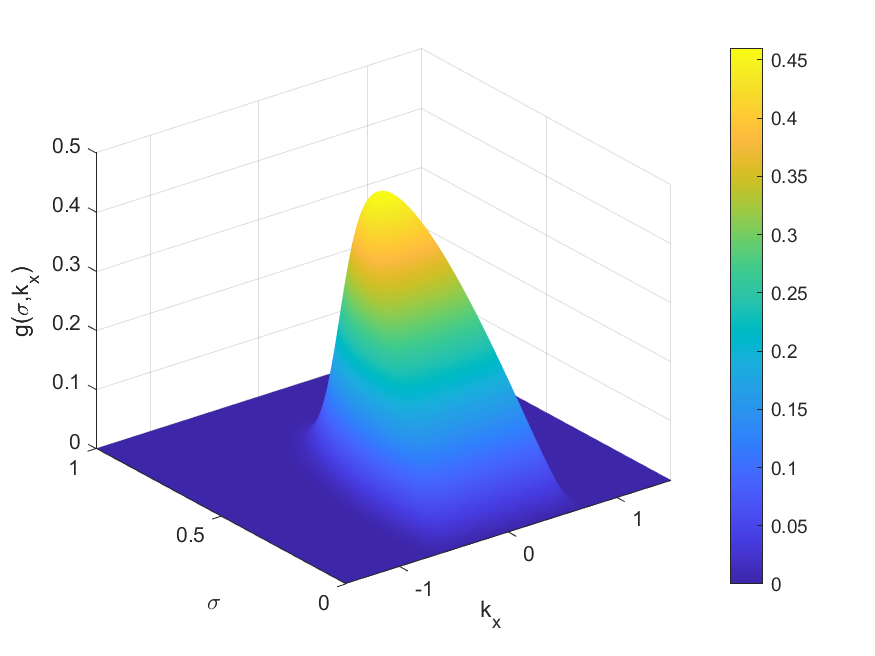}
       }
 \caption{Plot of the function $g (\sigma,k_x)$. One can note that numerically $g (\sigma,k_x)$ has compact support. }
    \label{fig:integrand}
\end{figure}

\begin{figure}[H]
    \centering
{
        \includegraphics[width=0.58\textwidth]{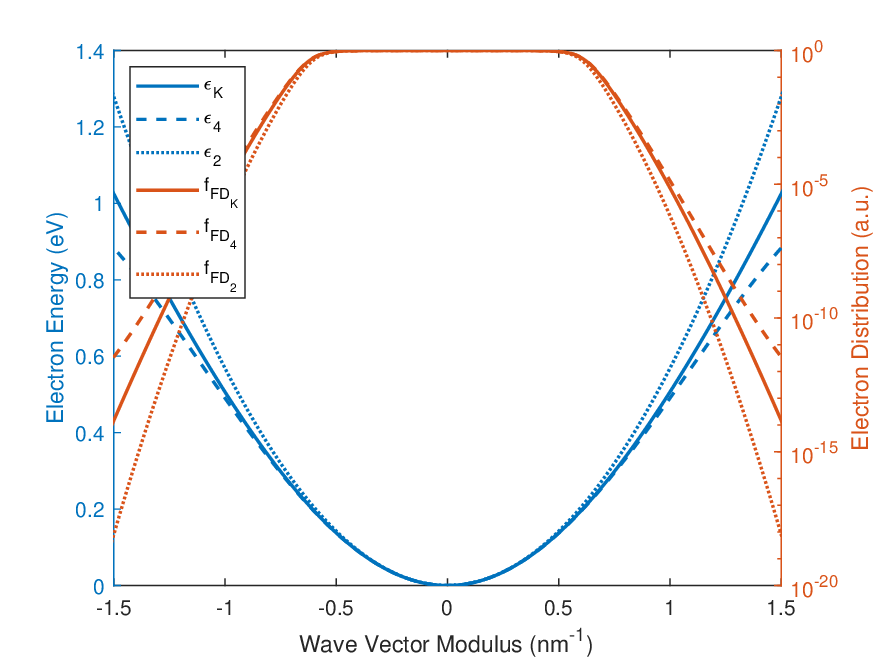}
       }
 \caption{Logarithmic plot of $f_{FD}(k)$ and the energy band with the full Kane dispersion relation and with the quartic and quadratic approximation. The notation is as in Fig.~\ref{fig:distribution}}
    \label{fig:FD_log}
\end{figure}

For $k_x > 0$, the quantity $J_{k_x}$ can be directly inferred from relation~\ref{trans_curr} evaluated at $x = L$. Similar considerations apply for $k_x < 0$.

Using Gaussian quadrature for numerical integration and solving the Schr\"odinger equation at each quadrature node with the scheme proposed in \cite{AliNaRo}, we obtain the results shown in Figures~\ref{fig:den4kane} and~\ref{fig:curr4kane}. Regarding the density while SE2 and SE4 yield qualitatively similar results, quantitative differences are evident: SE2 produces a higher peak in the resonant region for the electron density, whereas SE4 reveals interference effects. The density obtained with SE4 has only a slight difference between the case of the full Kane dispersion and the case of the quartic approximation.

The current is conserved with high numerical accuracy, confirming the robustness of the computational method. The values of the currents obtained with SE4 by adopting  the full Kane dispersion relation and the quartic approximation are about -6.53$\times 10^{-13}$ and -7.43 $\times 10^{-13}$ A/nm$^2$ while the current given by SE2 is about -1.71$\times 10^{-12}$ A/nm$^2$. SE4 yields a value of the current that is  approximately 38\% of that predicted by SE2 similarly to what expected in the semiclassical case where it is well known that the parabolic band overestimates the current (in absolute value). In the case of SE4 the difference in the current between the full Kane dispersion relation and the quartic approximation is to ascribe solely to the slight difference in the statistics.

\begin{figure}[H]
    \centering
{
        \includegraphics[width=0.48\textwidth]{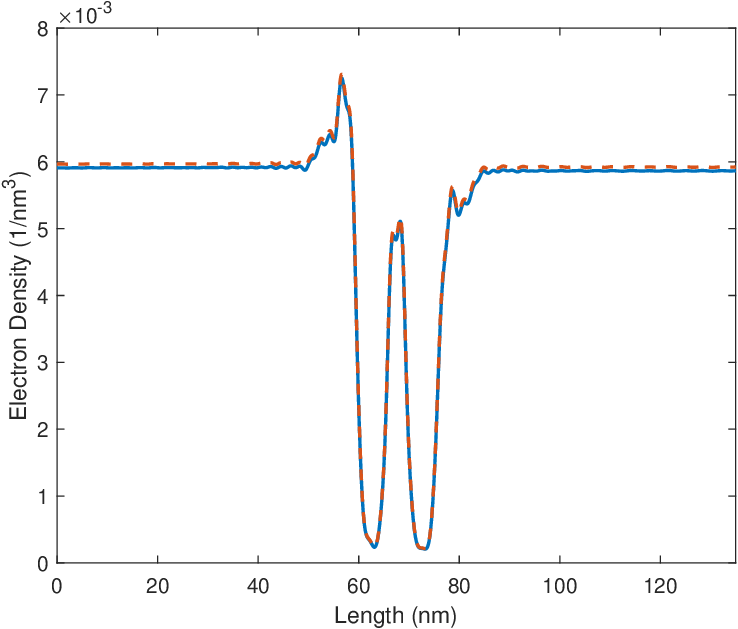} \includegraphics[width=0.48\textwidth]{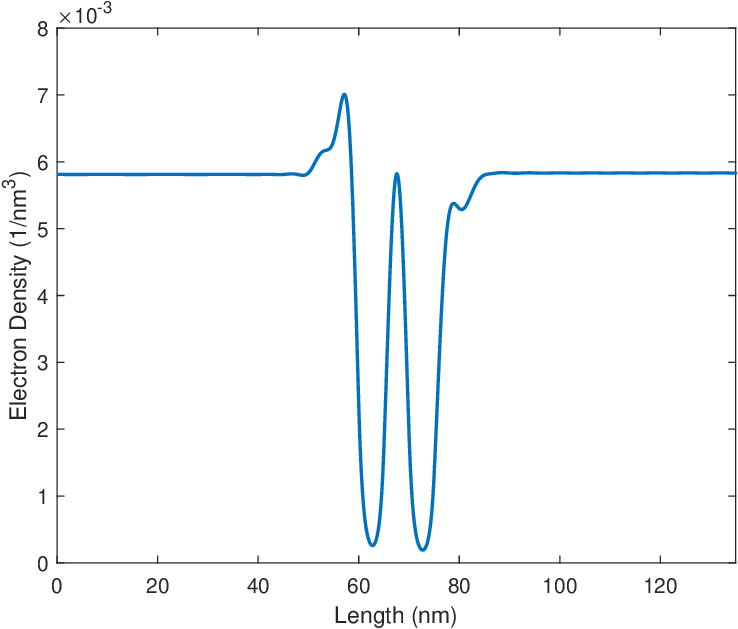}
       }
 \caption{Electron density obtained with the fourth-order SE (continuous line, left) , and the second order SE (right) in the case $V_0=0$ V, $V_L= 0.1$ V, $V_b = -0.3$ V,  $L=135$ nm by using the Kane dispersion relation. For comparison also the electron density with the quartic approximation in the statistics (dashed line, left) is reported}.
    \label{fig:den4kane}
\end{figure}

\begin{figure}[H]
    \centering
{
        \includegraphics[width=0.6\textwidth]{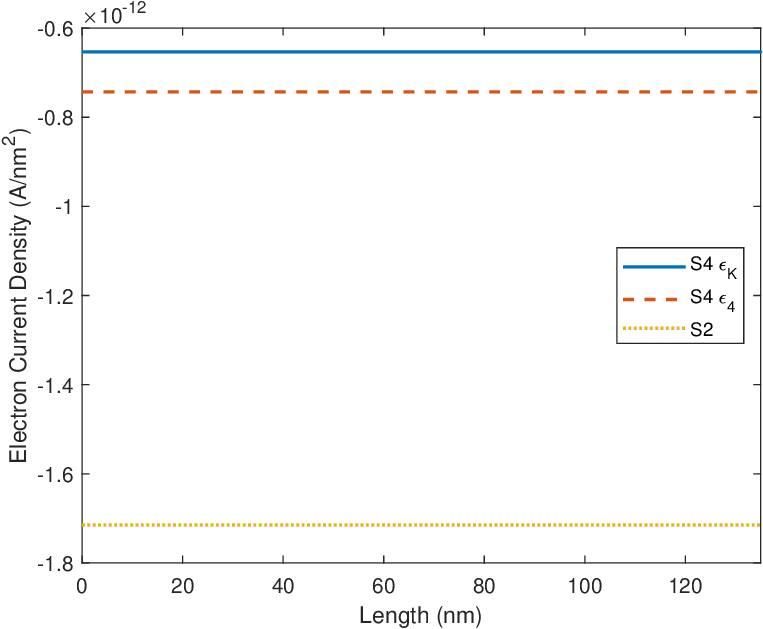} 
       }
 \caption{Electron current density obtained with the fourth-order SE by using the Kane dispersion relation and the second order SE  in the case $V_0=0$ V, $V_L= 0.1$ V, $V_b = -0.3$ V,  $L=135$ nm by using the Kane dispersion relation. For comparison also the electron current density with the quartic approximation in the statistics is reported.} 
    \label{fig:curr4kane}
\end{figure}

For completeness, we also compare results obtained using a parabolic band approximation in the Fermi-Dirac statistics. The corresponding solutions are shown in Figures~\ref{fig:denparab} and~\ref{fig:currparab}. The SE4 density exhibits more pronounced oscillations due to interference effects compared to SE2, and its magnitude is higher than in the Kane model. As in the semiclassical case, the parabolic band approximation tends to overestimate the electron current density in absolute value. SE2 again predicts a higher current than SE4, reinforcing the importance of employing accurate dispersion relations.

\begin{figure}[H]
    \centering
{
        \includegraphics[width=0.48\textwidth]{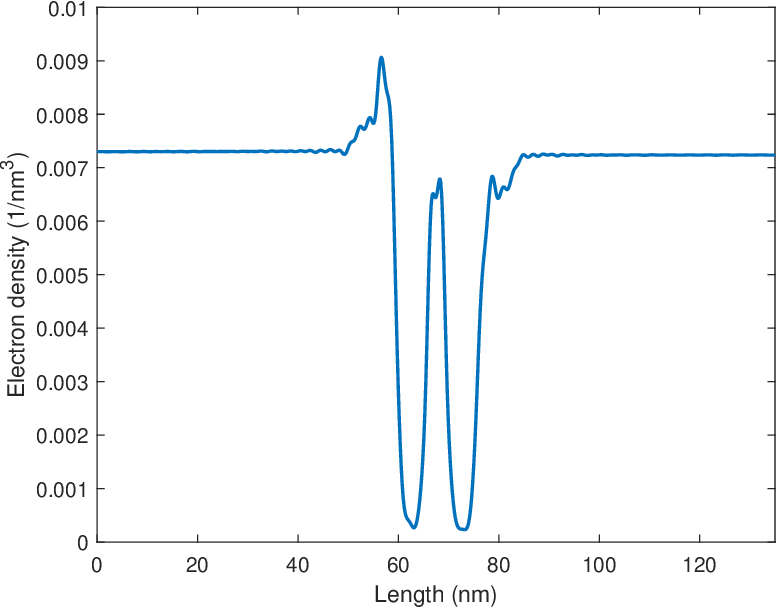} \includegraphics[width=0.48\textwidth]{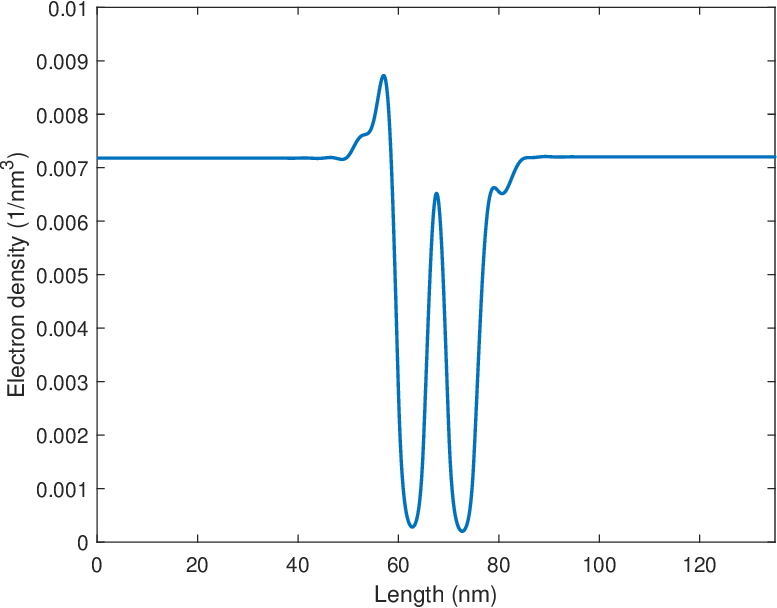}
       }
 \caption{Electron density in the case $V_0=0$ V, $V_L= 0.1$ V, $V_b = -0.3$ V,  $L=135$ nm in the of parabolic band approximation. SE4 left, SE2 right.}
    \label{fig:denparab}
\end{figure}

  \begin{figure}[H]
    \centering
{
        \includegraphics[width=0.7\textwidth]{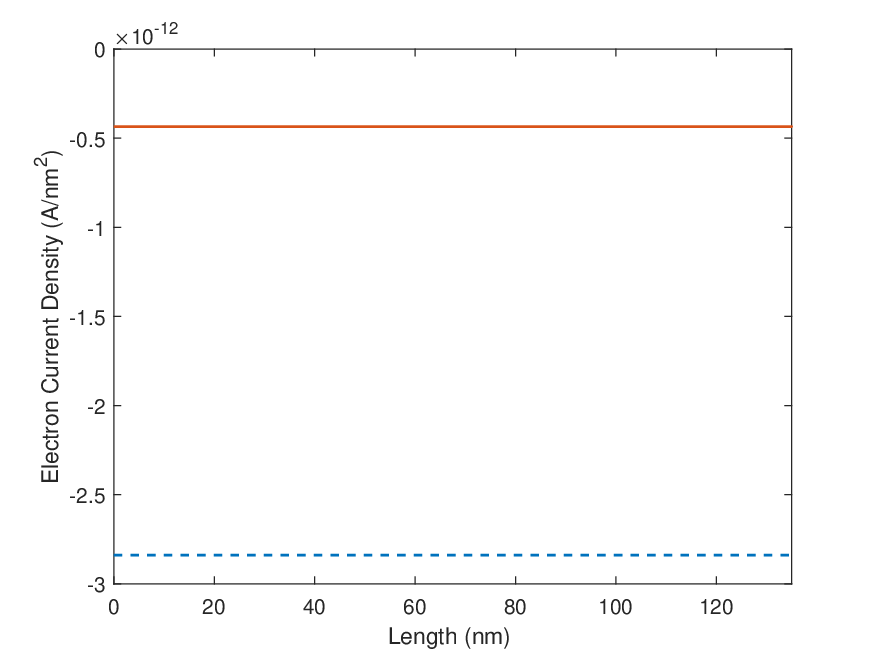} 
       }
 \caption{Electron current density in the case $V_0=0$ V, $V_L= 0.1$ V, $V_b = -0.3$ V,  $L=135$ nm in the parabolic band approximation. SE4: continuous line. SE2: dashed line.}
    \label{fig:currparab}
\end{figure}


\section*{Conclusion}

A general energy band structure has been incorporated into the Schr\"odinger equation in the coordinate representation to model charge transport in nanoscale devices. In particular, the Kane dispersion relation has been analyzed. By expanding in the non-parabolicity parameter, a hierarchy of Schr\"odinger equations of increasing order is obtained. The self-adjointness of the associated Hamiltonians has been rigorously established for electrostatic potentials relevant to nanoelectronic applications. Furthermore, suitable transparent boundary conditions have been formulated, enabling the definition of a boundary value problem over a finite spatial interval for the accurate description of charge transport in electron devices. A generalized expression for the electron current density has also been derived. Numerical simulations for resonant tunneling diodes (RTDs) highlight the capabilities of the proposed model and, in agreement with results from the semiclassical framework, emphasize the necessity of employing more precise dispersion relations to achieve improved evaluations of the electron current density.
\vskip 1cm

\noindent
{\bf Funding.} The authors acknowledge the support from INdAM (GNFM) and from MUR progetto PRIN
{\it Transport phonema in low dimensional structures: models, simulations and theoretical aspects} 
CUP E53D23005900006. V.R. acknowledges also the support from the project PIACERI LINEA 1 {\it TraPlas}, University of Catania.
\vskip 1cm
\noindent
{\bf Conflict of interest.} The authors declare that have not conflict of interest.

\clearpage \noindent

\end{document}